\documentclass[12pt]{article}

\usepackage[utf8]{inputenc}
\usepackage[T1]{fontenc}
\usepackage[margin=1.1in]{geometry}
\usepackage{setspace}
\usepackage{amsmath,amssymb,amsthm,amsfonts,mathtools}
\usepackage{mathrsfs}
\usepackage{enumitem}
\usepackage{booktabs}
\usepackage{hyperref}
\usepackage{natbib}

\hypersetup{
  colorlinks=true,
  linkcolor=blue,
  citecolor=blue,
  urlcolor=blue
}

\newtheorem{theorem}{Theorem}
\newtheorem{proposition}{Proposition}
\newtheorem{lemma}{Lemma}
\newtheorem{corollary}{Corollary}
\newtheorem{definition}{Definition}

\newtheorem{axiom}{Axiom}

\newcommand{\R}{\mathbb{R}}

\newcommand{\Prob}{\mathbb{P}}
\newcommand{\Q}{\mathcal{Q}}
\newcommand{\Qinf}{\mathcal{Q}_{\infty}}
\newcommand{\X}{\mathcal{X}}

\newcommand{\1}{\mathbf{1}}
\newcommand{\supp}{\operatorname{supp}}
\newcommand{\osc}{\operatorname{osc}}

\title{Spectral Aggregation of Quantile Preferences}
\author{Van-Quy Nguyen\thanks{Email: \url{quynv@neu.edu.vn}}\\[0.5em]
\small\parbox{0.9\textwidth}{\centering \it Faculty of Mathematical Economics, College of Technology, National Economics University, Hanoi, Vietnam.}}

\date{\today}

\begin{document}

\maketitle

\begin{abstract}
Many collective decisions under risk are made by people who care about different parts of the outcome distribution: downside losses, typical performance, or upside gains. This paper models this disagreement with quantile preferences and studies how the represented quantile levels can be aggregated. Our main result is a spectral support theorem: a spectral social aggregation satisfies the Pareto principle if and only if its social spectrum puts mass only on quantile levels represented in society. Hence, Pareto consistency makes representative-quantile aggregation a dictatorial case. In addition, we derive spectral aggregation from rank-based axioms, develop finite and threshold-Pareto consequences, and show when local benchmark-affine and elliptical common-shape domains admit a representative-quantile reduction.
\end{abstract}




\section{Introduction}
\label{sec:introduction}

Many important decisions under risk---such as pension boards selecting portfolios, regulators evaluating solvency, or public agencies determining infrastructure capacities---are made by committees rather than isolated expected-utility maximizers. In these collective environments, disagreement often extends beyond divergent beliefs or utility curvature; members frequently disagree on which specific rank of the outcome distribution should be prioritized. For instance, one member might focus on downside protection, while others emphasize typical performance or upside opportunities.

Quantile preferences provide a straightforward framework for modeling this form of heterogeneity. A decision maker characterized by a quantile level $\tau$ evaluates a random variable $X$ strictly by its quantile $Q_\tau[X]$. In this framework, low values of $\tau$ represent downside-oriented evaluation, central values correspond to median-type evaluation, and high values indicate upside-oriented evaluation. Since quantiles commute with continuous increasing transformations, the level $\tau$ itself carries the behavioral content: risk attitude is represented by the rank at which the distribution is read, rather than by the curvature of a utility. This intuition makes quantile levels natural objects to aggregate.

The central question is therefore direct: given individual quantile levels $\tau_1,\ldots,\tau_n$, which social evaluation rules respect Pareto comparisons? A tempting answer is to choose a single representative quantile level. This paper shows why that answer is generally too narrow. A representative level that no individual uses is a phantom level. Two acts can agree at every represented individual level while differing at this unrepresented level; the social ranking would then change even though no individual quantile evaluation changes.

The alternative studied here is spectral aggregation. A spectral social aggregation evaluates $X$ by averaging its quantile function with respect to a social spectrum $\mu$,
\[
  W_\mu(X)=\int_0^1 Q_u[X]d\mu(u).
\]
The social spectrum records how much social attention is assigned to each quantile level. Representative quantiles, finite weighted averages of individual quantiles, expected values, and tail-average criteria are all special cases. Spectral aggregation is therefore a minimal extension of representative-quantile aggregation that can track several represented ranks simultaneously.

Our main result is an exact support theorem. Weak Pareto holds for a spectral social aggregation if and only if the social spectrum is supported on the represented set $T(\boldsymbol\tau)=\{\tau_1,\ldots,\tau_n\}$. Strong Pareto holds if and only if, in addition, every represented level receives positive mass. Thus, Pareto consistency is equivalent to a no-phantom-level principle. It also explains the fragility of representative-quantile aggregation: weak Pareto forces the representative level to be one of the individual levels, while strong Pareto is possible only when all individuals have the same quantile level.

We also study three further implications. First, a small off-support mass creates only small Pareto-indifference violations, and any positive off-support mass can generate a violation of proportional size. Second, finite spectral aggregation provides a simple Pareto-compatible rule for heterogeneous committees, while a weaker threshold version of Pareto permits representative levels inside the interval spanned by the individual levels. Third, on local benchmark-affine domains, and in particular for portfolio committees with elliptical common-shape payoffs, the finite spectrum can sometimes be summarized by a single representative quantile. Outside such domains, the finite social spectrum remains the object directly justified by Pareto consistency.

The structure of the paper is as follows. Section \ref{sec:setup} introduces quantile preferences, social functional, and the rank-based axioms leading to spectral aggregation. Section \ref{sec:main-result} proves the Pareto support theorem and develops its consequences for representative quantile, finite spectral aggregation, and a weaken Pareto criteria. Section \ref{sec:local-domain-investment} studies local benchmark-affine domains and an application on the portfolio-committee under elliptical common-shape payoffs. Section \ref{sec:conclusion} gives a short conclusion and outlines further lines of research. All proofs are collected in the Appendix.

\subsection{Related literature}
\label{subsec:related-literature}
The first related strand is the theory of quantile preferences. While the concept of quantile maximization was initially introduced by \citet{Manski1988}, its axiomatic foundations were later established by \citet{Chambers2009}, who axiomatized quantiles as distributional functionals, and \citet{Rostek2010}, who axiomatized quantile-maximizing preferences. This framework was subsequently extended by \citet{deCastroGalvao2019} to dynamic models, and by \citet{deCastroGalvao2022} to general static and dynamic settings. Furthermore, a growing body of applied and experimental research demonstrates that quantile levels are crucial drivers of behavior in portfolio selection, laboratory settings, and firm investment \citep{deCastroGalvaoMontesRojasOlmo2022,deCastroGalvaoNoussairQiao2022,deCastroGalvaoKimMontesRojasOlmo2022,AlmeidaCampelloDeCastroGalvao2024}. Building on this foundational work, we take the individual quantile level as the primitive object of heterogeneity, shifting the focus to how these diverse levels can be systematically aggregated in society. 

The second strand addresses social aggregation under risk. The classical benchmark in this domain is \citet{Harsanyi1955}, which demonstrates that applying the Pareto principle in an expected-utility environment inevitably leads to an affine aggregation of individual utilities. In this paper, we obtain a similar result by weakening the Pareto principle. However, in general, since quantile preferences evaluate distributions locally at specific quantile levels, quantile aggregation can lead to a dictatorial aggregator, similar to \citet{Arrow1951} and the broader tradition of impossibility theorems.

Finally, this paper intersects with the broad literature on rank-based welfare and risk measurement. In decision theory, models such as rank-dependent utility, Yaari's dual theory, and Choquet expected utility evaluate acts by weighting outcomes according to their rank and utilizing nonadditive Choquet integration \citep{Yaari1987,Quiggin1982,Schmeidler1989}. In a parallel vein within welfare economics, generalized Gini social welfare functions apply rank-based weighting to positions in an ordered income distribution \citep{DonaldsonWeymark1980,Weymark1981}. Furthermore, in the mathematical finance literature, coherent risk measures provide the foundational axiomatic class for evaluating capital adequacy \citep{ArtznerEtAl1999}. Building on this coherence framework, spectral risk measures and Kusuoka representations characterize risk through weighted mixtures of quantile or tail-average functionals \citep{Acerbi2002,Kusuoka2001}. By synthesizing these diverse concepts, this paper bridges these fields, providing a novel social-choice interpretation of spectral objects tailored specifically to heterogeneous quantile preferences. 
\section{Quantile model}
\label{sec:setup}
\subsection{Quantile function}
Let $(\Omega, \mathcal{F}, \Prob)$ be an atomless probability space, and let $X$ be a real-valued random variable (an act) from $\Omega$ to $\R$, which can be interpreted as the payoff of the risky choice. In this paper, we focus on the set of bounded real-valued random variables, 
\[\X=L^\infty=\{X:X \text{ is a random variable taking values in }\R
  \text{ and } \|X\|_\infty<\infty\},\]
where $
  \|X\|_\infty =\operatorname*{ess\,sup}X$. With a
slight abuse of notation, for every $c \in \R$, we use $c$ to denote the constant act $X_c(\omega) = c$ for all $\omega \in \Omega$.

For a given  $X\in \X$,  its  (lower) quantile of $X$ is defined as 
\[
  Q_\tau[X]=
  \begin{cases}
  \inf\{x\in\R:F_X(x)\ge \tau\}, & \text{if } \tau\in(0,1],\\
  \sup\{x\in\R:F_X(x)=0\}, & \text{if } \tau=0,
  \end{cases}
\]
where $F_X$ is the cumulative distribution function of $X$, that is $F_X(x)=\Prob(X\le x)$. Note that,  at the endpoints, $Q_0[X]=\operatorname*{ess\,inf}X$ and $Q_1[X]=\operatorname*{ess\,sup}X$, so that the quantile assumes a value in the support of $X$. For convenience, throughout the paper we will focus on $\tau \in (0, 1)$, unless explicitly stated otherwise.

It's clear that, for every $X\in \X$, the quantile function $Q_\tau [X]$ belongs to the set of bounded, weakly increasing, and left-continuous functions on $(0,1)$, that is
\[
  \Qinf
  :=\{q:[0,1]\to\R:
  q \text{ is bounded, weakly increasing, and left-continuous on }(0,1)\}.
\]
Conversely, every $q\in\Qinf$ can be realized, up to the usual endpoint convention, by a bounded real-valued random variable $X \in \X$.

%
%
%

\begin{lemma}\label{lem:construction}
For every continuous weakly increasing $q:[0,1]\to\R$, there exists a bounded real-valued random variable $X\in\X$ such that
\[
  Q_u[X]=q(u)
  \qquad
  \forall u\in(0,1).
\]
\end{lemma}


\subsection{Quantile preferences}
\label{subsec:individual-ranks}
A well-known and important property of quantiles is its invariance with respect to monotonic transformations. More formally, if $\phi: \R \to \R$ is continuous and strictly increasing, then
\begin{equation}
Q_\tau [\phi(X)] = \phi(Q_\tau[X])
\label{quantile_transform}
\end{equation}
A preference $\succsim$ over random variables is a $\tau$-quantile preference (QP) for some fixed $\tau\in(0,1)$ if
\begin{equation}
X \succsim Y \iff Q_\tau[u(X)] \geq Q_\tau[u(Y)]
\label{quantile_pref_uti}
\end{equation}
where $u(\cdot)$ is the utility function which is continuous and strictly increasing.

It is worth noticing that, by (\ref{quantile_transform}), the QP defined by (\ref{quantile_pref_uti}) are independent of the utility function.\footnote{Indeed, $X \succsim Y \iff Q_\tau[u(X)] \geq Q_\tau[u(Y)] \iff u(Q_\tau[X]) \geq u(Q_\tau[Y]) \iff Q_\tau[X] \geq Q_\tau[Y]$.}
Therefore, the QP in (\ref{quantile_pref_uti}) can be defined as
\begin{equation}
X \succsim Y \iff Q_\tau[X] \geq Q_\tau[Y]
\label{quantile_pref}
\end{equation}
and the quantile level $\tau$ can be used to represent the taste or the behavior of individual with quantile preferences. Lower values of $\tau$ represent more risk-averse agents.\footnote{See Section 5 in \cite{deCastroGalvao2022} for more discussion.}

\subsection{Social aggregation of quantile levels}
\label{subsec:social-functionals}

A social aggregation is just a map $W:L^\infty \to\R$ that satisfies some \textit{social} properties later. We focus on a special class of social aggregation called spectral aggregation. That is, for a given Borel probability measure $\mu$ on $[0,1]$, define
\[
  W_\mu(X)=\int_0^1 Q_\tau[X]d\mu(\tau),
\]
whenever the integral is finite. The measure $\mu$ is the social spectrum. It records how much social attention is assigned to each quantile level of the outcome distribution. This spectral class contains many interesting classes. A representative quantile is the special case $\mu=\delta_{\tau_0}$; a finite weighted average of individual quantiles corresponds to a discrete measure; expected value and tail-average criteria are also spectral functionals under the usual integrability conditions.

More importantly, the spectral aggregation  is not
introduced merely as a convenient mathematical class. It has a natural behavior foundation, that is, it is  derived from a set of intuitive axioms below.
\begin{axiom}[Law invariance]\label{Law_inva} If $X$ and $Y$ have the same distribution, then $W(X)=W(Y)$.
\end{axiom}
\begin{axiom}[Statewise monotonicity] If $X\ge Y$ almost surely, then $W(X)\ge W(Y)$.
\label{State_mono}
\end{axiom}
\begin{axiom}[Constant invariance] For every constant $c\in\R$, $W(X+c)=W(X)+c$.
\label{Constant_inva}
\end{axiom}
\begin{axiom}[Comonotonic additivity] If $X$ and $Y$ are comonotonic, then $W(X+Y)=W(X)+W(Y)$.
\label{Como_add}
\end{axiom}
\begin{axiom}[Lower-quantile regularity on events] If $A_m\uparrow A$ and $\Prob(A)<1$, then $
    W(\1_{A_m})\longrightarrow W(\1_A)$.
    \label{Event_regu}
    \end{axiom}
Here $X$ and $Y$ are comonotonic if
\[
  (X(\omega)-X(\omega'))(Y(\omega)-Y(\omega'))\ge0
\]
for almost every pair $(\omega,\omega')$.\footnote{In the atomless setting, they may be represented as increasing transformations of a common random variable: there are a random variable $Z$ and weakly increasing functions $f,g$ such that $X=f(Z)$ and $Y=g(Z)$ almost surely.}

\begin{theorem}[Spectral representation]
\label{thm:axiomatic-spectral}
Let $W$ be social aggregation function satisfying Axioms \ref{Law_inva}-\ref{Event_regu} on the space of bounded real-valued random variables $\X=L^\infty$. Then there exists a unique Borel probability measure $\mu_W$ on $[0,1]$ such that
\[
  W(X)=\int_0^1 Q_u[X]d\mu_W(u)
\]
for every bounded real-valued random variable $X\in\X$.
\end{theorem}

\section{Aggregation results}
\label{sec:main-result}
There are $n$ individuals. Each individual $i$ has her own quantile level $\tau_i\in(0,1)$ and she evaluates random variables according to QP in (\ref{quantile_pref}), i.e.,
\[
  X\succsim_i Y
  \quad\Longleftrightarrow\quad
  Q_{\tau_i}[X]\ge Q_{\tau_i}[Y].
\]
A profile of quantile levels is denoted by $\boldsymbol\tau=(\tau_1,\dots,\tau_n)$, and the set of represented quantile levels is $T(\boldsymbol\tau)=\{\tau_1,\dots,\tau_n\}$. The set $T(\boldsymbol\tau)$ removes labels and multiplicities. Multiplicity matters later through the amount of social mass assigned to a quantile level, but Pareto support itself is a restriction on the set of quantile levels that are represented in society.

We recall the standard weak and strong Pareto principles for studying the aggregation of quantile levels.

\begin{definition}[Weak and strong Pareto]
\label{def:pareto}
A social functional $W$ satisfies \textbf{weak} Pareto at profile $\boldsymbol\tau$ if for every $X, Y \in \X$,
\[
  X\succsim_i Y
  \quad\forall i
  \quad\Longrightarrow\quad
  W(X)\ge W(Y).
\]
It satisfies \textbf{strong} Pareto at profile $\boldsymbol\tau$ if, whenever the weak inequalities hold and at least one of them is strict, then $W(X)>W(Y)$.
\end{definition}

A quantile level outside $T(\boldsymbol\tau)$ is called a phantom level for the profile. Such a level can enter a spectral social evaluation even though no individual reads the distribution there. The support results below show that Pareto consistency is exactly the condition that rules out phantom levels, with strong Pareto requiring positive social mass on every represented quantile level.

\subsection{Spectral Pareto support}
\label{sec:pareto-support}
\label{subsec:exact-support}

Having identified spectral aggregations as the relevant  quantile-linear social rules, we now ask which spectra are compatible with the individual quantile levels $(\tau_i)_{i=1,\ldots,n}$. The answer is a support condition. Weak Pareto allows social mass only on represented levels, while strong Pareto additionally requires every represented level to receive positive mass.

\begin{theorem}[Spectral Pareto support]
\label{thm:spectral-support}
Let $W_\mu$ be a spectral social aggregation with the social spectrum $\mu$. Then we have
\begin{enumerate}[label=(\roman*)]
  \item $W_\mu$ satisfies weak Pareto at $\boldsymbol\tau$ if and only if
  \[
    \mu\bigl([0,1]\setminus T(\boldsymbol\tau)\bigr)=0.
  \]
  \item $W_\mu$ satisfies strong Pareto at $\boldsymbol\tau$ if and only if it satisfies the support condition above and
  \[
    \mu(\{\alpha\})>0
    \qquad
    \forall \alpha\in T(\boldsymbol\tau).
  \]
\end{enumerate}
\end{theorem}

The exact support theorem is sharp, but applications often use estimated spectra. The following proposition quantifies the cost of placing small mass on phantom levels, turning the support theorem into a stability result. The off-support mass of $W_\mu$ at $\boldsymbol\tau$ is
\[
  \varepsilon_W(\boldsymbol\tau)
  :=\mu\bigl([0,1]\setminus T(\boldsymbol\tau)\bigr).
\]
\begin{proposition}[Approximate Pareto support]
\label{thm:approx-support}
Let $W_\mu$ be a spectral social functional with social spectrum $\mu$. Then the following statements hold.
\begin{enumerate}[label=(\roman*)]
  \item If $X,Y\in\X$ satisfy $Q_{\tau_i}[X]=Q_{\tau_i}[Y]$ for every $i$, then
  \[
    |W_\mu(X)-W_\mu(Y)|
    \le
    \varepsilon_W(\boldsymbol\tau)
    \max\{\osc(Q_\cdot[X]),\osc(Q_\cdot[Y])\},
  \]
  where $\osc(Q_\cdot[X])=\sup_{u\in[0,1]}Q_u[X]-\inf_{u\in[0,1]}Q_u[X]$.
  \item If $\varepsilon_W(\boldsymbol\tau)>0$, then for every $\eta>0$ there exist $X,Y\in\X$ such that $Q_{\tau_i}[X]=Q_{\tau_i}[Y]$ for every $i$, but
  \[
    |W_\mu(X)-W_\mu(Y)|
    \ge
    c\bigl(\varepsilon_W(\boldsymbol\tau)-\eta\bigr)
  \]
  for some constant $c>0$.
\end{enumerate}
\end{proposition}

Proposition \ref{thm:approx-support} identifies off-support mass as the maximal source of Pareto-indifference violations. Small off-support mass creates only small violations, while positive off-support mass always permits a violation of proportional size after a small perturbation.

\subsection{Consequences of spectral support}
This subsection gives the main finite-population consequences of the spectral support theorem containing a representative social function and weighted-linear social function.

We begin with single representative social function, where social aggregation function is represented by $ W(X)=Q_{\tau_0}[X]$ for some $\tau_0 \in (0,1)$. The following result is a direct consequence of Theorem \ref{thm:spectral-support} with $\mu=\delta_{\tau_0}$ where $\delta_{\tau_0}$ is a Dirac measure.

\begin{corollary}\label{cor:single-rank}
Suppose society is represented by a single quantile level $\tau_0$.

\begin{enumerate}[label=(\roman*)]
  \item Weak Pareto holds if and only if $\tau_0\in T(\boldsymbol{\tau})$.
    \item Strong Pareto holds if and only if $    T(\boldsymbol{\tau})=\{\tau_0\}$.
\end{enumerate}
\end{corollary}

Corollary \ref{cor:single-rank} identifies the central failure of representative aggregation. Weak Pareto does not merely require the social quantile level to lie between the most pessimistic and most optimistic individuals. It requires society to select one of the individual quantile levels exactly. Strong Pareto is even more restrictive: it requires all individual quantile levels coincide.

Another natural alternative of social function is the  weighted-linear aggregation, in which we aggregate several quantile levels simultaneously. More precisely, there are weights $\lambda_i \geq 0$ and $\sum_{i=1}^n \lambda_i=1$ such that
\[
  W_{\boldsymbol{\lambda}}(X)
  =
  \sum_{i=1}^n \lambda_i Q_{\tau_i}[X].
\]
For each realized quantile level $\alpha\in T(\boldsymbol{\tau})$, define the total weight assigned to $\alpha$ by
\[
  \Lambda(\alpha\mid\boldsymbol{\tau})
  =
  \sum_{i:\tau_i=\alpha}\lambda_i.
\]
The following corollary is directly obtained by taking   $\mu=\sum_{i=1}^n \lambda_i\delta_{\tau_i}$ in Theorem \ref{thm:spectral-support}.

\begin{corollary}
\label{cor:finite-multirank}
The weighted-linear functional $W_{\boldsymbol{\lambda}}$ satisfies weak Pareto. It satisfies strong Pareto if and only if
\[
  \Lambda(\alpha\mid\boldsymbol{\tau})>0
  \qquad
  \forall \alpha\in T(\boldsymbol{\tau}).
\]
\end{corollary}

This aggregation function resolves the problem faced by representative aggregation. Instead of choosing one social quantile level, society places weight on the levels actually represented among the individuals. The only Pareto restriction is that strong Pareto requires every realized quantile level to receive positive total mass.

\subsection{A possibility result}
\label{subsec:threshold-pareto}

Theorem \ref{thm:spectral-support} relies on full pairwise Pareto comparisons between arbitrary acts $X$ and $Y$. This condition is quite demanding; in a representative case, it leads to a dictatorial aggregator.  In this subsection, we study weaker Pareto conditions that no longer force society to select a dictator.

\begin{definition}[C-Pareto]
\label{def:threshold-pareto}
A social functional $W:\X\to\R$ satisfies \textbf{C-Pareto with lower threshold} at $\boldsymbol\tau$ if, for every $X\in\X$ and every constant $c\in\R$,
\[
  X \succsim_i c \quad \forall i
  \qquad\Longrightarrow\qquad
  W(X)\ge c.
\]
It satisfies \textbf{C-Pareto with upper threshold} at $\boldsymbol\tau$ if, for every $X\in\X$ and every constant $c\in\R$,
\[
  c\succsim_i X \quad \forall i
  \qquad\Longrightarrow\qquad
  c\ge W(X).
\]
It satisfies \textbf{C-Pareto} if it satisfies both conditions.
\end{definition}

These weaker conditions are naturally interpreted as acceptance-set conditions in the theory of coherent risk measures. In that literature, a risk measure is often studied in terms of the set of positions regarded as acceptable, and the standard coherence axioms impose structural restrictions on these acceptance sets \citep{ArtznerEtAl1999}. The present condition is weaker and more ordinal: it does not require subadditivity or positive homogeneity. It only requires that the social investor's cash-acceptance set contain the intersection of the experts' cash-acceptance sets. Moreover, deciding whether an action is acceptable relative to a constant act is typically easier than comparing two arbitrary acts. Thus, this new Pareto principle is less restrictive than the standard Pareto principle above.\footnote{For related constant-act Pareto and caution conditions in expert aggregation, see Section 3 of \cite{bachdong24}.}

This weaker form of Pareto requires only that the social quantile remain within the interval spanned by the individual quantiles, i.e., in the interval $[\underline\tau,\overline\tau]$ where $\underline\tau:=\min_{1\le i\le n}\tau_i$ and $\overline\tau:=\max_{1\le i\le n}\tau_i$.

\begin{theorem}\label{thm:threshold-support-interval}
Let $W_\mu$ be a spectral social aggregation with the social spectrum $\mu$. Then we have
\begin{enumerate}[label=(\roman*)]
  \item $W_\mu$ satisfies C-Pareto with lower threshold at $\boldsymbol\tau$ if and only if $
    \mu([\underline\tau,1])=1.$
  \item $W_\mu$ satisfies C-Pareto with upper threshold at $\boldsymbol\tau$ if and only if $
    \mu([0,\overline\tau])=1.$ 
  \item $W_\mu$ satisfies C-Pareto at $\boldsymbol\tau$ if and only if $ 
    \mu([\underline\tau,\overline\tau])=1.$
\end{enumerate}
\end{theorem}

In the financial market, each expert provides a quantile-based acceptance test for investments. The two C-Pareto conditions require the investor to respect expert judgments: if all experts accept (or reject) an investment against a sure payoff, so must the investor. The lower-threshold condition rules out social mass below $\underline\tau$, so the investor cannot be more pessimistic than the most pessimistic expert. The upper-threshold condition rules out social mass above $\overline\tau$, so the investor cannot be more optimistic than the most optimistic expert. Hence, the investor's quantile spectrum should lie between the least and greatest expert quantile levels. Within this interval, the investor may smooth, average, or otherwise distribute attention.

Theorem \ref{thm:threshold-support-interval} is also related to Kusuoka's representation of law-invariant coherent risk measures. In Kusuoka's representation, a law-invariant coherent risk measure satisfying the Fatou property is represented as a supremum over mixtures of tail-value-at-risk, or tail-average, functionals (see Theorem 4 in \cite{Kusuoka2001}). Kusuoka's result asks: given coherence and law invariance, what representation does a risk measure have? The present theorem asks instead: given expert quantile judgments and Pareto, where may the investor's quantile levels lie? The answer is not that all coherent levels are allowed. The investor's spectrum must be inside the interval of expert quantile levels: $\operatorname{supp}(\mu)\subseteq [\underline\tau,\overline\tau]$. Thus, Theorem \ref{thm:threshold-support-interval} is different but complementary to Kusuoka's result.
  
\section{Local Domains and Investment Applications}
\label{sec:local-domain-investment}

The previous spectral support theorem is intentionally a rich-domain result: when the risky choice $X$ can vary freely, the Pareto principle forces representative aggregation to select from the represented levels. This section studies the case in which the domain of risky choices is only \textit{locally} restricted. More precisely, the relevant comparisons are concentrated near particular decision-relevant candidates, rather than ranging over the entire distribution.

This local perspective is useful in many applied environments. For example, a bank may care mainly about losses near a regulatory capital or value-at-risk threshold; an insurer may focus on claims near a solvency cutoff; and a university or employer may compare applicants close to an admission or hiring margin. In such cases, alternatives need not share a global location--scale structure. It is enough that they align locally around the thresholds or candidates that drive the decision. We provide one concrete application in Subsection~\ref{subsec:portfolio-committees}, where, under elliptical returns, a Pareto-supported committee spectrum can be implemented as a single mean--standard-deviation investment mandate.

\subsection{A local benchmark-affine domain}
\label{subsec:local-z-affine}
Let $Z$ be a benchmark random variable with strictly increasing continuous quantile function $z(u)=Q_u[Z]$ for $u\in(0,1)$. Fix a set of relevant quantile levels $\mathcal T\subset(0,1)$. A family $\mathcal D$ of strictly increasing continuous quantile functions is called \emph{$Z$-local on $\mathcal T$} if it contains all linear transformations of $z$, i.e., $\{a+bz(\cdot):a\in\R,\ b>0\}$, and for every $q\in\mathcal D$, there exist $a_q\in\R$ and $b_q>0$ such that $q(\tau)=a_q+b_qz(\tau)$ for all $\tau\in \mathcal T$. The associated local domain is 
\[\X_{\mathcal D}=\{q(U):q\in\mathcal D\}, \quad U\sim\operatorname{Unif}(0,1).\]
This local domain contains all linear transformations of benchmark $Z$, but it also contains quantile curves with different skewness, tail thickness, or curvature away from $\mathcal T$. What is fixed is only the benchmark-affine geometry at the quantile coordinates that matter for the Pareto principle. This is why, on this local domain, the social quantile level is not a dictatorial level, but uniquely determined by all represented quantile levels.

\begin{theorem}[Local-domain spectral aggregation]
\label{thm:local-domain-spectral}
Fix a profile $\boldsymbol\tau=(\tau_1,\dots,\tau_n)$ and let $W_\mu$ be a spectral social aggregation with the social spectrum $\mu$. Suppose $\mu((0,1))=1$ and $\mathcal D$ is $Z$-local on a set $\mathcal T= T(\boldsymbol\tau)\cup\supp(\mu)\subset(0,1)$. Then, the following conditions are equivalent:
\begin{enumerate}[label=(\roman*)]
  \item $W_\mu$ satisfies weak Pareto at $\boldsymbol\tau$ on $\X_{\mathcal D}$;
  \item $W_\mu$ satisfies C-Pareto at $\boldsymbol\tau$ on $\X_{\mathcal D}$;
  \item $ \kappa_\mu
    \in
    \operatorname{co}\{z(\alpha):\alpha\in T(\boldsymbol\tau)\}$, where  $\kappa_\mu =  \int_{(0,1)}z(u)\,d\mu(u)$.
\end{enumerate}
Moreover, $W_\mu$ satisfies strong Pareto at $\boldsymbol\tau$ on $\X_{\mathcal D}$ if and only if
\[
  \kappa_\mu\in
  \operatorname{ri}\!\left(\operatorname{co}\{z(\alpha):\alpha\in T(\boldsymbol\tau)\}\right).
\]
\end{theorem}
The support of social spectrum $\mu$ identifies the quantile levels that enter the spectral aggregation, while the benchmark $Z$ gives these levels a common payoff scale. Hence $\kappa_\mu$ is the benchmark location of the social spectrum: it summarizes where the spectral aggregation lies on the benchmark quantile curve. In particular, $  \kappa_\mu \in \operatorname{co}\{z(u):u\in\supp(\mu)\}$. The Pareto conditions in Theorem~\ref{thm:local-domain-spectral} then say that this benchmark location must lie inside the benchmark interval spanned by the represented quantile levels, namely $\operatorname{conv}\{z(\alpha):\alpha\in T(\boldsymbol\tau)\}$. Thus Pareto consistency on the local domain does not require the social spectrum to be supported only on represented quantile levels; it requires the benchmark location of the social spectrum to fall within the benchmark-affine span of the represented quantile levels.

More interestingly, if $\kappa_\mu\in z((0,1))$, define $\tau_\mu^Z=F_Z(\kappa_\mu)$, then
\[
  W_\mu(X)=Q_{\tau_\mu^Z}[X]
  \qquad
  \forall X\in\X_{\mathcal D},
\]
provided that $\mathcal D$ is $Z$-local on $\mathcal T\cup\{\tau_\mu^Z\}$. Thus, on the local domain, a general spectral rule can be  represented by a representative rule. 

We end this subsection by providing a corollary directly obtained from Theorem~\ref{thm:local-domain-spectral} by taking $\mu=\delta_{\tau_0}$, for which $\kappa_\mu=z(\tau_0)$. 

\begin{corollary}\label{cor:local-domain-pareto}
Fix a profile $\boldsymbol\tau=(\tau_1,\dots,\tau_n)$, let $\tau_0\in(0,1)$ be a social quantile level. Suppose $\mathcal D$ is $Z$-local on $\mathcal T=T(\boldsymbol\tau)\cup\{\tau_0\}$. Then the representative rule $W_0(X)=Q_{\tau_0}[X]$ satisfies weak Pareto (or C-Pareto) at $\boldsymbol\tau$ on $\X_{\mathcal D}$ if and only if $ z(\tau_0)\in \operatorname{co}\{z(\alpha):\alpha\in T(\boldsymbol\tau)\}$. Consequently, there are weights $\lambda_i\ge0$ with $\sum_i\lambda_i=1$ such that 
\[  \tau_0=F_Z\!\left(\sum_{i=1}^n\lambda_iQ_{\tau_i}[Z]\right).
\]
Moreover, $W_0$ satisfies strong Pareto if and only if $  z(\tau_0)
  \in
  \operatorname{ri}\!\left(\operatorname{co}\{z(\alpha):\alpha\in T(\boldsymbol\tau)\}\right)$.
\end{corollary}

Corollary~\ref{cor:local-domain-pareto} establishes that local-domain Pareto consistency generates a Kolmogorov--Nagumo quasi-arithmetic mean on quantile levels. Rather than averaging social quantile levels directly, the aggregation procedure works in three steps: transform to the benchmark scale, average there, then transform back. The benchmark distribution $Z$ thus serves as the generator of social compromise. It equips the unit interval of quantile levels with a benchmark-affine geometry, enabling Pareto-consistent compromise to be linear on the benchmark quantile curve—not on the probability scale.

\subsection{Portfolio committees and investment mandates}
\label{subsec:portfolio-committees}

This subsection applies the local-domain theorem to a portfolio committee. Each member is represented by a quantile preference at a quantile level $\tau_i$, while the institution chooses portfolio weights from a feasible mandate set. The purpose is to show that, when feasible portfolio payoffs have an elliptical common-shape structure, the benchmark-affine geometry in Subsection~\ref{subsec:local-z-affine} becomes the usual mean--scale geometry of portfolio theory. The committee's finite social spectrum can then be written as a single representative quantile rule and, equivalently, as a mean--standard-deviation mandate.

\paragraph{Committee aggregation.}
Let $R$ be the vector of asset returns or terminal payoffs, and let the committee choose portfolio weights $w$ from a feasible set $B$. The set $B$ may encode budget balance, no-short-sale restrictions, leverage limits, tracking-error constraints, liquidity requirements, capital requirements, or institutional screens. The payoff generated by $w$ is
\[
  X_w=w'R.
\]

Individual $i$ evaluates $X_w$ by the quantile functional $Q_{\tau_i}[X_w]$. Lower quantile levels, such as $0.05$, emphasize bad states and can be interpreted as solvency or shortfall criteria. Central quantile levels, such as $0.50$, represent typical performance. Upper quantile levels, such as $0.90$, emphasize upside states. Let $\lambda_i\ge0$ with $\sum_i\lambda_i=1$. These weights may represent voting weights, fiduciary shares, assets represented, or the anonymous benchmark $\lambda_i=1/n$.

The induced finite spectral social aggregation is
\[
  W_{\mu_\lambda}(X_w)
  =
  \sum_{i=1}^n\lambda_iQ_{\tau_i}[X_w],
  \qquad
  \mu_\lambda=
  \sum_{i=1}^n\lambda_i\delta_{\tau_i}.
\]
For notational simplicity, write $W_\lambda(w):=W_{\mu_\lambda}(X_w)$. Since $\mu_\lambda$ is supported on $T(\boldsymbol\tau)$, Corollary~\ref{cor:finite-multirank} implies weak Pareto. If every represented quantile level receives positive total weight, then the same corollary gives strong Pareto.

\paragraph{Elliptical common-shape structure.}
We now impose a portfolio-domain restriction. This restriction is stronger than specifying only a mean vector and covariance matrix. It requires a common elliptical generator, so that changing $w$ may change the location and scale of $X_w$, but not the standardized shape of its quantile function. Formally, suppose
\[
  R\sim EC_m(\bar R,\Sigma,\psi),
\]
meaning that the characteristic function of $R$ has the form
\[
  \varphi_R(t)=\exp(i t'\bar R)\,\psi(t'\Sigma t),
  \qquad t\in\R^m,
\]
for a common generator $\psi$. Then every linear payoff satisfies
\[
  X_w\overset d= w'\bar R+\sqrt{w'\Sigma w}\,Z_0,
  \qquad w\in B,
\]
where the scalar benchmark shock $Z_0$ does not depend on $w$. If $w'\Sigma w=0$, the payoff is degenerate and the formula is interpreted with zero scale. We assume that $F_{Z_0}$ is continuous and strictly increasing and that $Q_{\tau_i}[Z_0]$ is finite for every represented quantile level.

This common-shape condition is the linear-projection property of elliptically contoured distributions. The distributional foundations are given by \citet{Kelker1970} and \citet{CambanisHuangSimons1981}; standard monograph treatments include \citet{FangKotzNg1990} and \citet{GuptaVargaBodnar2013}. In portfolio theory, the condition is closely related to the mean--variance and elliptical-distribution literature, including \citet{Chamberlain1983}, \citet{OwenRabinovitch1983}, and \citet{Berk1997}. Related formulas for tail risk under elliptical distributions are provided by \citet{LandsmanValdez2003}, with broader risk-management treatments in \citet{McNeilFreyEmbrechts2015}.

Let $z_0(u)=Q_u[Z_0]$. For every interior quantile level $u$,
\[
  Q_u[X_w]
  =
  w'\bar R+
  \sqrt{w'\Sigma w}\,z_0(u).
\]
Thus the feasible portfolio domain is contained in the location--scale family generated by the benchmark $Z_0$. In particular, it is $Z_0$-local on every finite set of relevant quantile levels. Hence the local-domain result in Subsection~\ref{subsec:local-z-affine} applies directly.

\begin{proposition}[Elliptical committee reduction]
\label{thm:committee-elliptical-reduction}
Assume the elliptical common-shape structure above. For committee weights $\lambda_i\ge0$ with $\sum_i\lambda_i=1$, define
\[
  \kappa_\lambda
  =
  \int_{(0,1)} z_0(u)\,d\mu_\lambda(u)
  =
  \sum_{i=1}^n\lambda_iQ_{\tau_i}[Z_0],
  \qquad
  \tau_\lambda=F_{Z_0}(\kappa_\lambda).
\]
Then, for every feasible portfolio $w\in B$,
\[
  W_\lambda(w)
  =
  \sum_{i=1}^n\lambda_iQ_{\tau_i}[X_w]
  =
  Q_{\tau_\lambda}[X_w]
  =
  w'\bar R+
  \kappa_\lambda\sqrt{w'\Sigma w}.
\]
Consequently, maximizing the Pareto-supported finite spectrum over $B$ is equivalent to solving the single mandate
\[
  \max_{w\in B}
  \left\{w'\bar R+
  \kappa_\lambda\sqrt{w'\Sigma w}\right\}.
\]
\end{proposition}

The coefficient $\kappa_\lambda$ is the benchmark location of the committee's social spectrum, in the same sense as $\kappa_\mu$ in Theorem~\ref{thm:local-domain-spectral}. It is determined by the represented quantile levels and the committee weights. The market inputs $\bar R$, $\Sigma$, and the benchmark distribution $Z_0$ enter separately. When $\kappa_\lambda<0$, the mandate penalizes portfolio scale; when $\kappa_\lambda>0$, it rewards upside scale; and when $\kappa_\lambda=0$, the aggregation reduces to mean payoff on the elliptical domain.

\paragraph{Normal benchmark.}
If $Z_0$ is standard normal with distribution function $\Phi$, then
\[
  \kappa_\lambda=
  \sum_i\lambda_i\Phi^{-1}(\tau_i),
  \qquad
  \tau_\lambda=
  \Phi(\kappa_\lambda).
\]
For three equally weighted committee members, the conversion is
\begin{center}
\begin{tabular}{lccc}
\toprule
committee quantile levels & $\kappa_\lambda$ & $\tau_\lambda$ & mandate interpretation \\
\midrule
$(0.05,0.10,0.50)$ & $-0.976$ & $0.165$ & downside-oriented \\
$(0.05,0.50,0.90)$ & $-0.121$ & $0.452$ & nearly balanced, mildly defensive \\
$(0.50,0.75,0.90)$ & $0.652$ & $0.743$ & upside-oriented \\
\bottomrule
\end{tabular}
\end{center}

The table illustrates that, under the normal benchmark, committee aggregation is linear on the normal-score scale rather than on the probability scale: with equal weights, $\tau_\lambda=\Phi\!\left(\frac{1}{3}\sum_{i=1}^3\Phi^{-1}(\tau_i)\right)$. Hence the resulting mandate is governed by the sign and magnitude of $\kappa_\lambda$. The profile $(0.05,0.10,0.50)$ gives $\kappa_\lambda=-0.976$ and $\tau_\lambda=0.165$, producing a downside-oriented mandate that penalizes portfolio scale; $(0.05,0.50,0.90)$ gives $\kappa_\lambda=-0.121$ and $\tau_\lambda=0.452$, producing a nearly balanced but mildly defensive mandate; and $(0.50,0.75,0.90)$ gives $\kappa_\lambda=0.652$ and $\tau_\lambda=0.743$, producing an upside-oriented mandate that rewards scale. Thus the representative level is not the arithmetic average of committee quantile levels, but a benchmark-scale compromise that exactly implements the finite spectral rule under the elliptical common-shape condition.

\section{Conclusion and further research}
\label{sec:conclusion}

This paper develops a social-choice framework for aggregating quantile preferences. The main message is a support principle for social spectra: weak Pareto permits social mass only on the represented quantile levels, and strong Pareto additionally requires that every represented level receive positive mass. This no-phantom-level principle makes representative-quantile aggregation a fragile special case: a single social quantile can satisfy Pareto only by selecting one represented level. Finite spectral aggregation is therefore the natural Pareto-supported rule for heterogeneous quantile levels. The local-domain results show that, on special local domains,  the finite spectrum can be compressed into a single representative quantile.

Further research can move in several directions. One direction is to combine quantile-level aggregation with heterogeneous beliefs while keeping rank aggregation and belief pooling conceptually distinct. A second is the study of strategic elicitation, when individuals can misreport their quantile levels. Finally, relaxing the elliptical common-shape restriction would clarify when a representative quantile is an accurate reduction and when the finite Pareto-supported spectrum is indispensable.

\appendix
\section{Appendix}
\label{app:proofs}

\subsection{Proof of Lemma \ref{lem:construction}}
\begin{proof}[Proof of Lemma \ref{lem:construction}]
Let $U$ be uniformly distributed on $(0,1)$ and define $X=q(U)$. Fix $u\in(0,1)$. If $x<q(u)$, continuity and monotonicity imply that $q(v)>x$ on some left-neighborhood of $u$, and monotonicity also gives $q(v)>x$ for all $v\ge u$. Hence
\[
  \Prob(X\le x)=\Prob(q(U)\le x)<u.
\]
If $x>q(u)$, then $q(v)<x$ for every $v\le u$, so
\[
  \Prob(X\le x)\ge \Prob(U\le u)=u.
\]
Therefore the lower quantile $\inf\{x:\Prob(X\le x)\ge u\}$ is exactly $q(u)$. Since $u$ was arbitrary, $Q_u[X]=q(u)$ for all $u\in(0,1)$.
\end{proof}

\subsection{Proof of Theorem \ref{thm:axiomatic-spectral}}
\begin{proof}[Proof of Theorem \ref{thm:axiomatic-spectral}]
We use a direct Riesz-representation argument on smooth quantile curves and then pass to arbitrary bounded quantile curves by step approximation. The only nonsmooth step uses the event-regularity axiom on tail indicators.

Fix a uniformly distributed random variable $U$ on $(0,1)$, which exists because the probability space is atomless.

Comonotonic additivity applied to the pair $(0,0)$ gives $W(0)=0$. Hence constant invariance gives $W(c)=c$ for all $c\in\R$. Statewise monotonicity and constant invariance imply that $W$ is one-Lipschitz in the sup norm. Indeed, if $\|X-Y\|_\infty\le\varepsilon$, then $Y-\varepsilon\le X\le Y+\varepsilon$, and therefore
\[
  W(Y)-\varepsilon
  =W(Y-\varepsilon)
  \le W(X)
  \le W(Y+\varepsilon)
  =W(Y)+\varepsilon .
\]
Thus, one has $|W(X)-W(Y)|\le \|X-Y\|_\infty$.

Comonotonic additivity and this continuity also imply positive homogeneity. For integers $n\ge1$, repeated comonotonic additivity gives $W(rY)=rW(Y)$ for all rational number $r\geq 0$. The one-Lipschitz property extends this identity to every real $r\ge0$.

Now define the set of weakly increasing continuous functions on $[0,1]$ as
\[
  \Q:=C^1_\uparrow([0,1])
  =\{q\in C^1[0,1]:q \text{ is weakly increasing}\}
\]
and for $q\in\Q$, define $\mathcal W(q):=W(q(U))$.

Law invariance makes this definition well-defined. If $q\ge r$ pointwise, then $q(U)\ge r(U)$ almost surely, so monotonicity gives $\mathcal W(q)\ge\mathcal W(r)$. Also, its's easy to see that $\mathcal W$ satisfies constant invariance, i.e., $\mathcal W(q+c\1)=\mathcal W(q)+c$ for all $c\in\R$ where $\1$ is the constant function takes value $1$. If $q,r\in\Q$, then $q(U)$ and $r(U)$ are comonotonic, since both are increasing functions of the same random variable $U$. Therefore, $\mathcal W$ satisfies additivity, i.e., $\mathcal W(q+r)=\mathcal W(q)+\mathcal W(r)$.

Now, every $f\in C^1[0,1]$ can be written as a difference of two elements of $\Q$: for example, if $M>\|f'\|_\infty$ and $e(u)=u$, then $f=(f+Me)-Me$, and both $f+Me$ and $Me$ are weakly increasing. And if $f=q-r=q'-r'$ with $q,r,q',r'\in\Q$, then by the additivity of $\mathcal W$, one has
\[
  \mathcal W(q)-\mathcal W(r)
  =\mathcal W(q')-\mathcal W(r'),
\]
so the following definition is well-defined:
\[
  L(f):=\mathcal W(q)-\mathcal W(r)
  \qquad \text{whenever } f=q-r,\ q,r\in\Q.
\]
Using above linearity properties of $\mathcal W$, one shows that $L$ is additive on $C^1[0,1]$, and $L$ agrees with $\mathcal W$ on $\Q$. Using one-Lipchitz property of $\mathcal{W}$, one has
\[
  |L(f)|\le \|f\|_\infty .
\]
Thus $L$ is sup-norm continuous. Since $L$ is additive, this continuity gives homogeneity for all real scalars, so $L$ is linear. Because $C^1[0,1]$ is dense in $C[0,1]$, $L$ extends uniquely to a continuous positive linear functional on $C[0,1]$.

By the Riesz representation theorem, there exists a unique finite Borel measure $\mu_W$ on $[0,1]$ such that
\[
  L(f)=\int_0^1 f(u)\,d\mu_W(u)
  \qquad \forall f\in C[0,1].
\]
It's easy to see that $L$ is positive and $L(\1)=1$. Thus, $\mu_W$ is positive and $\mu_W([0,1])=L(\1)=1$. Hence $\mu_W$ is a Borel probability measure. Since $L(q)=\mathcal W(q)$ for $q\in\Q$, we obtain
\[
  W(q(U))=\int_0^1 q(u)\,d\mu_W(u)
  \qquad \forall q\in\Q.
\]
By uniform approximation of continuous increasing functions by smooth increasing functions, together with the one-Lipschitz property of $W$, the same identity holds for every continuous weakly increasing $q:[0,1]\to\R$. 

The next step identifies the value of $W$ on upper-tail indicators. Fix $s\in(0,1)$ and set $ A_s:=\{U>s\}$. For $\varepsilon>0$ with $s+\varepsilon<1$, define the continuous increasing function
\[
  \phi_{s,\varepsilon}(u)=
  \begin{cases}
    0, & u\le s,\\[1mm]
    \dfrac{u-s}{\varepsilon}, & s<u<s+\varepsilon,\\[2mm]
    1, & u\ge s+\varepsilon .
  \end{cases}
\]
Then $\1_{A_{s+\varepsilon}}
  \le \phi_{s,\varepsilon}(U)
  \le \1_{A_s}$, which implies
\[
  W(\1_{A_{s+\varepsilon}})
  \le W(\phi_{s,\varepsilon}(U))
  \le W(\1_{A_s}).
\]
As $\varepsilon\downarrow0$, the events $A_{s+\varepsilon}$ increase to $A_s$, and $\Prob(A_s)=1-s<1$. Lower-quantile regularity on events therefore gives $W(\1_{A_{s+\varepsilon}})\xrightarrow[\varepsilon\downarrow 0]{} W(\1_{A_s})$. Thus, $W(\phi_{s,\varepsilon}(U))\xrightarrow[\varepsilon\downarrow 0]{} W(\1_{A_s})$. Using bounded convergence, one has
\[
  \int_0^1 \phi_{s,\varepsilon}(u)\,d\mu_W(u)
  \longrightarrow
  \mu_W((s,1]).
\]
Consequently, $W(\1_{A_s})=\mu_W((s,1])$ for all $s\in(0,1)$. It is then straightforward to see that for every increasing step quantile curve
\[
  q(u)=x_0+\sum_{k=1}^m a_k\1_{(s_k,1]}(u),
  \qquad 0<s_k<1,\quad a_k\ge0,
\]
we have $W(q(U))=\int_0^1 q(u)\,d\mu_W(u)$. Thus the representation holds for all finite increasing step quantile curves.

The last step is to extend the formula to arbitrary bounded random variables. Let $X\in L^\infty$ and write $q(u)=Q_u[X]$ for  $0\le u\le1$. By the usual quantile representation, $q(U)$ has the same distribution as $X$. Hence law invariance gives $W(X)=W(q(U))$.

Let $a=q(0)$ and $b=q(1)$. If $a=b$, then $X=a$ almost surely and the formula follows from constant invariance. Suppose $a<b$. For $m\ge1$, set $\Delta_m=(b-a)/m$ and define
\[
  q_m(u)
  =a+\Delta_m\sum_{\ell=1}^{m}
       \1_{\{q(u)>a+\ell\Delta_m\}}.
\]
Because $q$ is a lower-quantile curve, each nonempty set $ \{u:q(u)>a+\ell\Delta_m\}$ is an upper-tail interval $(s_{\ell,m},1]$ with $s_{\ell,m}\in(0,1)$ whenever the threshold $a+\ell\Delta_m$ lies strictly between $a$ and $b$; empty terms are omitted. Hence each $q_m$ is a finite increasing step quantile curve, which follows $W(q_m(U))=\int_0^1 q_m(u)\,d\mu_W(u)$.

Moreover, $ q_m(u)\le q(u)\le q_m(u)+\Delta_m$ for all $u\in[0,1]$ gives 
\[W(q_m(U))\xrightarrow[m \to \infty]{} W(q(U)).\]
And since $q_m\to q$ uniformly and $\mu_W$ is a probability measure, 
\[\int_0^1 q_m(u)\,d\mu_W(u) \xrightarrow[m \to \infty]{}
  \int_0^1 q(u)\,d\mu_W(u).\]
Thus, 
\[
  W(X)=W(q(U))=\int_0^1 Q_u[X]d\mu_W(u).
\]

Finally, uniqueness follows from the Riesz representation theorem.
\end{proof}

\subsection{Proof of Theorem \ref{thm:spectral-support}}
\begin{proof}[Proof of Theorem \ref{thm:spectral-support}]
Assume first that $\mu$ is supported on $T(\boldsymbol\tau)$. Then
\[
  W_\mu(X)-W_\mu(Y)
  =\sum_{\tau\in T(\boldsymbol\tau)}\mu(\{\tau\})
  \bigl(Q_\tau[X]-Q_\tau[Y]\bigr).
\]
If every individual weakly prefers $X$ to $Y$, every term in this finite sum is nonnegative. Hence weak Pareto holds.

Conversely, suppose $\mu([0,1]\setminus T(\boldsymbol\tau))>0$. Define
\[
  h(u)=-\prod_{\tau\in T(\boldsymbol\tau)}(u-\tau)^2.
\]
Then $h$ vanishes at every represented quantile level and is strictly negative outside $T(\boldsymbol\tau)$. Since $\mu$ assigns positive mass outside $T(\boldsymbol\tau)$, we have $\int h(u)d\mu(u)<0$. 

Now consider two random variable $X$ and $Y$ with associated quantile curves $q_Y(u)=u$ and $q_X(u)=u+\varepsilon h(u)$. They exist because for sufficiently small $\varepsilon>0$, $q_X$ remains increasing. All individuals $i$ are indifferent because the perturbation vanishes at their quantile levels. But
\[
  W_\mu(X)-W_\mu(Y)=\varepsilon\int h(u)d\mu(u)<0.
\]
Thus weak Pareto fails. This proves the weak-Pareto equivalence.

For strong Pareto, first suppose that $\mu([0,1]\setminus T(\boldsymbol\tau))>0$. Fix $\beta\in T$ and consider 
\[h_\beta(u)=\prod_{\tau\in T(\boldsymbol\tau)\setminus\{\beta\}}(u-\tau)^2,\]
with the empty product interpreted as one. For sufficiently small $a>0$, $\int (h+a h_\beta)d\mu<0$. 

The perturbation $h+a h_\beta$ is zero at all represented levels except $\beta$, where it is positive. For small $\varepsilon>0$, there exists a random variable $X_\beta$ with its quantile curve $q_{X_\beta}(u)=u+\varepsilon (h(u)+ah_\beta (u))$. All individuals with level $\beta$ strictly prefer $X_\beta$ to $Y$, all other individuals are indifferent, but $W_\mu(X_\beta)-W_\mu(Y)<0$. Hence strong Pareto fails whenever off-support mass is positive.

Now assume that $\mu$ is supported on $T(\boldsymbol\tau)$ but that $\mu(\{\beta\})=0$ for some represented level $\beta\in T(\boldsymbol\tau)$. Choose the same $h_\beta$ as above. For small $\varepsilon>0$, $q_X(u)=u+\varepsilon h_\beta(u)$ is increasing. Individuals at level $\beta$ strictly prefer $X$ to $Y$, and all others are indifferent. But the social value is unchanged because $\beta$ receives zero mass and $h_\beta$ vanishes on $T(\boldsymbol\tau)\setminus\{\beta\}$. Strong Pareto fails.

If, conversely, $\mu$ is supported on $T$ and every represented level has positive mass, any strict individual improvement creates at least one strictly positive term in the finite spectral sum, while all other terms are nonnegative. Hence the social improvement is strict.
\end{proof}

\subsection{Proof of Proposition \ref{thm:approx-support}}
\begin{proof}[Proof of Proposition \ref{thm:approx-support}]
Let $A=[0,1]\setminus T(\boldsymbol\tau)$, and write  $ \varepsilon_W=\mu_W(A)$. By Theorem \ref{thm:axiomatic-spectral}, one has
\[
  W(X)-W(Y)
  =\int_0^1\bigl(Q_u[X]-Q_u[Y]\bigr)d\mu_W(u).
\]

\textit{(i)} Assume $X,Y\in\X$ and $Q_{\tau_i}[X]=Q_{\tau_i}[Y]$ for every $i$.  Since the integrand is zero on the finite set $T(\boldsymbol\tau)$,
\[
  W(X)-W(Y)
  =\int_A\bigl(Q_u[X]-Q_u[Y]\bigr)d\mu_W(u).
\]
Choose any represented rank $\tau\in T(\boldsymbol\tau)$. Let us denote  $\max\{\osc(Q_\cdot[X]),\osc(Q_\cdot[Y])\}=B$, then it's clear that both quantile functions are increasing with oscillation at most $B$. So, for every $u\in[0,1]$, we have
\[
  |Q_u[X]-Q_u[Y]|\le B.
\]
Therefore
\[
  |W(X)-W(Y)|
  \le \int_A |Q_u[X]-Q_u[Y]|d\mu_W(u)
  \le B\mu_W(A)
  =B\varepsilon_W.
\]
\textit{(ii)} Assume $\varepsilon_W>0$. For every $\eta>0$ one can find a compact set $K\subset A$ such that $ \mu_W(K)\ge \varepsilon_W-\eta$. 
Since $T(\boldsymbol\tau)$ is finite and disjoint from $K$, there exists a continuously differentiable function $h:[0,1]\to[0,1]$ such that $h=0$ on $T(\boldsymbol\tau)$ and $h=1$ on $K$. Hence
\[
  \int_0^1 h(u)d\mu_W(u)\ge \varepsilon_W-\eta.
\]
Define $ m_h=\max\left\{0,-\inf_{u\in[0,1]}h'(u)\right\}$ and choose $\gamma=c(m_h+1)$. By constructing, the random variables $\tilde X=q_{\tilde X}(U)$ and $\tilde Y=q_{\tilde Y}(U)$ belong to the main domain $\X$ where $q_{\tilde X}(u)=\gamma u+ch(u)$ and $q_{\tilde Y}(u)=\gamma u$.
Because $h$ vanishes on $T(\boldsymbol\tau)$, $Q_{\tau_i}[\tilde X]=Q_{\tau_i}[\tilde Y] $ for all $i$. However,
\[
  W(\tilde X)-W(\tilde Y)
  =\int_0^1\bigl(q_{\tilde X}(u)-q_{\tilde Y}(u)\bigr)d\mu_W(u)
  =c\int_0^1 h(u)d\mu_W(u)
  \ge c(\varepsilon_W-\eta).
\]
This proves the converse.
\end{proof}

%

\subsection{Proof of Theorem \ref{thm:threshold-support-interval}}
\begin{proof}[Proof of Theorem \ref{thm:threshold-support-interval}]
We prove the lower threshold part; the upper threshold part is symmetric. Suppose first that $\mu([\underline\tau,1])=1$. If $  Q_{\tau_i}[X]\ge c$ for all $i$, then in particular $Q_{\underline\tau}[X]\ge c$. Since the quantile function is increasing in its rank argument, $Q_u[X]\ge c$ for all $u\ge \underline\tau$. Therefore
\[
  W_\mu(X)=\int_{[\underline\tau,1]}Q_u[X]d\mu(u)\ge c.
\]
Thus C-Pareto with lower threshold is true.

Conversely, suppose $\mu([0,\underline\tau))>0$. Then there exists $a<\underline\tau$ such that $\mu([0,a])>0$. Choose a bounded increasing quantile curve $q$ satisfying
\[
  q(u)=-M \quad \text{on }[0,a],
  \qquad
  q(u)=0 \quad \text{on }[\underline\tau,1],
\]
with a monotone interpolation on $(a,\underline\tau)$. WLOG, taking $c=0$ and let $X=q(U)$ for a uniform random variable $U$.  Then every individual satisfies $Q_{\tau_i}[X]\ge0$, while, for $M$ large enough,
\[
  W_\mu(X)=\int_0^1q(u)d\mu(u)<0=c.
\]
This violates lower-threshold C-Pareto. Hence lower-threshold C-Pareto implies the support condition $\mu([\underline\tau,1])=1$.

For C-Pareto with upper threshold, the same argument with a positive bump above $\overline\tau$ shows that the condition is equivalent to $\mu([0,\overline\tau])=1$. Combining the two equivalences gives the two-sided statement.
\end{proof}

\subsection{Proof of Theorem \ref{thm:local-domain-spectral}}
\begin{proof}[Proof of Theorem \ref{thm:local-domain-spectral}]
Let $I_S=\operatorname{co}\{z(\alpha):\alpha\in T(\boldsymbol\tau)\}$.
Since $\mathcal D$ is $Z$-local on $\mathcal T=T(\boldsymbol\tau)\cup\supp(\mu)$, for every $q\in\mathcal D$ there are $a_q\in\R$ and $b_q>0$ such that
\[
  q(u)=a_q+b_qz(u)
  \qquad \forall u\in T(\boldsymbol\tau)\cup\supp(\mu).
\]
Therefore, for $X=q(U)$,
\[
  W_\mu(X)
  =\int q(u)d\mu(u)
  =a_q+b_q\kappa_\mu,
  \qquad
  \kappa_\mu=\int z(u)d\mu(u).
\]
Take $X=q(U)$ and $Y=r(U)$, the individual weak-Pareto inequalities are
\[
  (a_q-a_r)+(b_q-b_r)z(\alpha)\ge0
  \qquad \forall \alpha\in T(\boldsymbol\tau),
\]
and the social inequality is
\[
  (a_q-a_r)+(b_q-b_r)\kappa_\mu\ge0.
\]
If $\kappa_\mu\in I_S$, then $\kappa_\mu$ is a convex combination of the numbers $z(\alpha)$ with $\alpha\in S$, so the individual inequalities imply the social inequality. Conversely, if $\kappa_\mu\notin I_S$, there is an affine function $x\mapsto A+Bx$ that is nonnegative on $I_S$ but negative at $\kappa_\mu$. Choose $D>|B|$, set $r(u)=Dz(u)$ and $q(u)=A+(D+B)z(u)$, and note that both curves belong to $\mathcal D$ because $\mathcal D$ contains all benchmark-affine curves with positive scale. Their difference is $A+Bz(\cdot)$, so they give a Pareto violation on $\X_{\mathcal D}$. Thus weak Pareto is equivalent to $\kappa_\mu\in I_S$.

For C-Pareto, compare $X=q(U)$ with a constant $c$. The lower threshold condition requires that
\[
  a_q+b_qz(\alpha)\ge c \quad \forall\alpha\in S
  \quad\Longrightarrow\quad
  a_q+b_q\kappa_\mu\ge c,
\]
for all $a_q\in\R$, $b_q>0$ and $c \in \R$. This is equivalent to $\kappa_\mu\ge\min_{\alpha\in S}z(\alpha)$. The upper threshold condition is equivalent to $\kappa_\mu\le\max_{\alpha\in S}z(\alpha)$. Hence two-sided C-Pareto is equivalent to $\kappa_\mu\in I_S$, the same condition as weak Pareto.

Finally, suppose $\kappa_\mu\in\operatorname{ri}(I_S)$. If the individual inequalities hold and at least one is strict, then the affine function $x\mapsto A+Bx$ is nonnegative on $I_S$ and positive at some represented point. An affine function with this property is strictly positive at every point in the relative interior of $I_S$. Hence $A+B\kappa_\mu>0$, giving strong Pareto.

Conversely, suppose $\kappa_\mu\notin\operatorname{ri}(I_S)$. If $\kappa_\mu\notin I_S$, choose an affine function $x\mapsto A+Bx$ that is nonnegative on $I_S$, strictly positive at some represented point, and negative at $\kappa_\mu$. If $\kappa_\mu\in I_S\setminus\operatorname{ri}(I_S)$, choose a supporting affine function that is nonnegative on $I_S$, zero at $\kappa_\mu$, and positive at some represented point. Using the same benchmark-affine construction as above gives individual weak improvements with at least one strict improvement, but the social improvement is nonpositive. Strong Pareto therefore fails. This proves the stated strong-Pareto equivalence.
\end{proof}


\subsection{Proof of Proposition \ref{thm:committee-elliptical-reduction}}
\begin{proof}[Proof of Proposition \ref{thm:committee-elliptical-reduction}]
Under the elliptical common-shape assumption, every feasible linear payoff satisfies
\[
  X_w\overset d= w'\bar R+\sqrt{w'\Sigma w}\,Z_0.
\]
Therefore, for every represented quantile level $\tau_i$,
\[
  Q_{\tau_i}[X_w]
  =w'\bar R+\sqrt{w'\Sigma w}\,Q_{\tau_i}[Z_0].
\]
Multiplying by $\lambda_i$ and summing gives
\[
  W_\lambda(w)
  =w'\bar R+
  \sqrt{w'\Sigma w}\sum_{i=1}^n\lambda_iQ_{\tau_i}[Z_0]
  =w'\bar R+
  \kappa_\lambda\sqrt{w'\Sigma w}.
\]
Because $\kappa_\lambda$ is a convex combination of the benchmark quantiles $Q_{\tau_i}[Z_0]$ and $F_{Z_0}$ is continuous and strictly increasing, $\tau_\lambda=F_{Z_0}(\kappa_\lambda)$ is well defined and
\[
  Q_{\tau_\lambda}[Z_0]=\kappa_\lambda.
\]
Thus
\[
  Q_{\tau_\lambda}[X_w]
  =w'\bar R+
  \kappa_\lambda\sqrt{w'\Sigma w}
  =W_\lambda(w).
\]
The maximization equivalence follows because the three displayed objective functions coincide pointwise on the feasible set $B$.
\end{proof}

\bibliographystyle{apalike}

\end{document}